\documentclass[conference]{IEEEtran}
\IEEEoverridecommandlockouts
\usepackage{array}
\usepackage{epsfig}
\usepackage{cite}
\usepackage{epstopdf}
\usepackage{amsfonts,amsmath,amsthm,amssymb,balance}
\usepackage{multirow,balance}
\usepackage{bm}
\usepackage[usenames,dvipsnames]{color}
\usepackage{colortbl}
\usepackage{float}
\usepackage[ruled,vlined,linesnumbered]{algorithm2e}
\usepackage{graphicx}
\usepackage{thmtools}
\usepackage{subcaption}
\usepackage{stfloats}
\usepackage{url}
\usepackage{easyReview}

\definecolor{mygray}{gray}{.9}

\graphicspath{{figure/}}
\newcolumntype{C}[1]{>{\PreserveBackslash\centering}p{#1}}
\newcolumntype{R}[1]{>{\PreserveBackslash\raggedleft}p{#1}}
\newcolumntype{L}[1]{>{\PreserveBackslash\raggedright}p{#1}}
\newcommand{\AlgoResetCount}{\renewcommand{\@ResetCounterIfNeeded}{\setcounter{AlgoLine}{0}}}
\newcommand{\AlgoNoResetCount}{\renewcommand{\@ResetCounterIfNeeded}{}}
\newcounter{AlgoSavedLineCount}

\theoremstyle{definition}
\newtheorem{theorem}{Theorem}

\SetKwInput{KwInput}{Input}              
\SetKwInput{KwOutput}{Output}        
\SetKwRepeat{Do}{do}{while}

\graphicspath{{figure/}}

\def\BibTeX{{\rm B\kern-.05em{\sc i\kern-.025em b}\kern-.08em
    T\kern-.1667em\lower.7ex\hbox{E}\kern-.125emX}}

\begin{document}

\title{Secrecy Outage Probability Fairness for Intelligent Reflecting Surface-Assisted Uplink Channel}

\author{
	{	Xiangrui Cheng\IEEEauthorrefmark{1}, Yiliang Liu\IEEEauthorrefmark{2}, Zhou Su\IEEEauthorrefmark{2}, and Wei Wang\IEEEauthorrefmark{3}}
	\\
	\IEEEauthorblockA{
		\IEEEauthorrefmark{1}School of Software Engineering, Xi'an Jiaotong University, Xi'an 710049, China. \\
		\IEEEauthorrefmark{2}School of Cyber Science and Engineering, Xi'an Jiaotong University, Xi'an 710049, China. \\
		\IEEEauthorrefmark{3}College of Electronic and Information Engineering, Nanjing University of Aeronautics and Astronautics, Nanjing, China.\\
		Email: innsdcc@stu.xjtu.edu.cn, liuyiliang@xjtu.edu.cn, zhousu@ieee.org, wei\_wang@nuaa.edu.cn}
	%\IEEEauthorrefmark{1}author.one@add.on.net,
	%\IEEEauthorrefmark{2}author.two@add.on.net,
	%\IEEEauthorrefmark{3}author.three@add.on.net,
	%\IEEEauthorrefmark{4}author.four@add.on.net}
}
%NUAA 211106
%\author{
%    \IEEEauthorblockN{ Xiangrui Cheng, Yiliang Liu,  Zhou Su, and Wei Wang}
%        \IEEEauthorblockA{School of Software Engineering, Xi'an Jiaotong University, Xi'an, China\\
%        School of Cyber Science and Engineering, Xi'an Jiaotong University, Xi'an, China.\\
%        College of Electronic and Information Engineering, Nanjing University of Aeronautics and Astronautics, Nanjing, China.\\
%        Email:innsdcc@stu.xjtu.edu.cn,liuyiliang@xjtu.edu.cn,zhousu@ieee.org,wei\_wang@nuaa.edu.cn
%        }
%% % \thanks{X. Cheng (email: {\tt innsdcc@stu.xjtu.edu.cn}), Y. Liu (email: {\tt liuyiliang@xjtu.edu.cn}), and Z. Su (email: {\tt zhousu@ieee.org}) are with the School of Cyber Science and Engineering, Xi'an Jiaotong University, Xi'an, Shaanxi, China. W. Zhang (email: {\tt wei\_wang@nuaa.edu.cn}) is with the College of Electronic and Information Engineering, Nanjing University of Aeronautics and Astronautics, Nanjing, China. H.-H. Chen (email: {hshwchen@mail.ncku.edu.tw}) is with the Department of Engineering Science, National Cheng Kung University, Tainan, Taiwan.
%% %  }
%}

\maketitle

\begin{abstract}
This paper investigates physical layer security (PLS) in the intelligent reflecting surface (IRS)-assisted multiple-user uplink channel. Since the instantaneous eavesdropper's channel state information (CSI) is unavailable, the secrecy rate can not be measured. In this case, existing investigations usually focus on the maximization of the minimum (max-min) of signal to interference plus noise power ratio (SINRs) among multiple users, and do not consider secrecy outage probability caused by eavesdroppers. In this paper, we first formulate the minimization of the maximum (min-max) secrecy outage probability among multiple users. The formulated problem is solved by alternately optimizing receiving matrix and phase shift matrix. Simulations demonstrate that the maximum secrecy outage probability is significantly reduced with the proposed algorithm compared to max-min SINR strategies, meaning our scheme has a higher security performance.

\end{abstract}

%\begin{IEEEkeywords}
%physical layer security, intelligent reflecting surface, secrecy outage probability, min-max fairness
%\end{IEEEkeywords}

\section{Introduction}
Recently, intelligent reflecting surfaces (IRS) have attracted much attention due to their excellent potential for the development of 6G networks \cite{Gong2020}. As a special meta surface with plenty of reflecting elements, it has the capability of controlling the scattering, reflection, and refraction characteristics of the radio waves. A large body of research has been motivated by IRS to enlarge coverage and improve spectral/energy-efficiencies for wireless communications \cite{Nguyen2022,Zheng2022}. 

On account of increasing security threats of wireless communication, IRS has also been used to strengthen physical layer security (PLS)  \cite{Liu2022,Sai2022,Dong2021,Asaad2022,Wang2021}. Many IRS-assisted PLS schemes have been conducted to focus on the maximization of secrecy rate by optimizing the phase shift matrix of IRS \cite{Sai2022,Dong2021,Asaad2022}. However, secrecy rate or sum secrecy rate of multiple users maximization methods have the following problems. Firstly, the instantaneous CSIs of eavesdroppers are usually unknown, so the secrecy rate can not be measured. It requires deriving the secrecy outage probability \cite{Liu2022} to replace the secrecy rate as the optimization objective. Secondly, the sum secrecy rate maximization methods \cite{Asaad2022} are not suitable for the multiple user uplink channel because they do not consider the fairness of multiple users. Optimization algorithms being so fully engaged in sum secrecy rate maximization may provide some users with very low secrecy capacities or very large secrecy outage probabilities. Hence, a fairness mechanism should be considered in the optimization design. Due to the unknown eavesdropper's instantaneous CSI, the researchers in \cite{Wang2021, Gong2021MMSINR, Nadeem2020MMSINR } presented the max-min SINR strategy to improve the minimum main channel capacity among multiple users. Nevertheless, the max-min SINR strategy aimed at traditional communication fairness is not designed for minimizing secrecy outage probability, so secrecy outage probability fairness schemes for IRS need more in-depth research.

To solve these problems mentioned above, we formulate a min-max fairness problem in terms of secrecy outage probabilities in the multiple-user uplink channel, where the maximum secrecy outage probability among users can be reduced with the IRS-assisted scheme. To our knowledge, this is the first time the issue has been investigated. The main contributions of this work are summarized as follows.

\begin{itemize}
\item We formulate a multiple-user fairness problem in an IRS-assisted uplink channel to minimize the maximum secrecy outage probability among users. 

\item Since the formulated problem is non-convex and hard to solve, we first transform the original problem into the optimization problems of receiving matrix and phase shift matrix, and solve them via generalized Rayleigh quotient and generalized Dinkelbach's algorithm. Then, we propose an alternating optimization (AO) algorithm to find the global results for the original problem. 
\end{itemize}

The paper is organized as follows. The system model and problem formulation are introduced in Section \ref{sec:sysmodel_and_problem}. Then, the alternating optimization algorithm is proposed in Section \ref{sec:Opt}. Simulation results are provided in Section \ref{sec:simu} and conclusions are presented in Section \ref{sec:conclusions}.

\textsl{Notations:}
Bold uppercase letters, such as $\mathbf{A}$, denote matrices, and bold lowercase letters, such as $\mathbf{a}$, denote column vectors. $\mathbf{A}^{\dagger}$, $\mathbf{A}^{\rm{T}}$, and $\mathbf{A}^{\rm{H}}$ represent the conjugate transformation, transpose, and conjugate transpose of $\mathbf{A}$, respectively. $\mathbf{I}_a$ is an identity matrix with its rank $a$. $\mathcal{CN}(\mu,\sigma^2)$ is a complex normal (Gaussian) distribution with mean $\mu$ and variance $\sigma^2$. $(\mathbf{A})^{-1}$ is the inverse function of $\mathbf{A}$. $|\mathbf{x}|$ is the Euclidean norm of $\mathbf{x}$. $\text{diag}(\mathbf{x})$ is the diagonal matrix of $\mathbf{x}$. $\text{vec}(\mathbf{A})$ is the vectorization of the diagonal elements of $\mathbf{A}$.

\section{System Model and Problem Formulation}  \label{sec:sysmodel_and_problem}
\subsection{System Model}
\begin{figure} 
\centering
\includegraphics[width=0.85\linewidth]{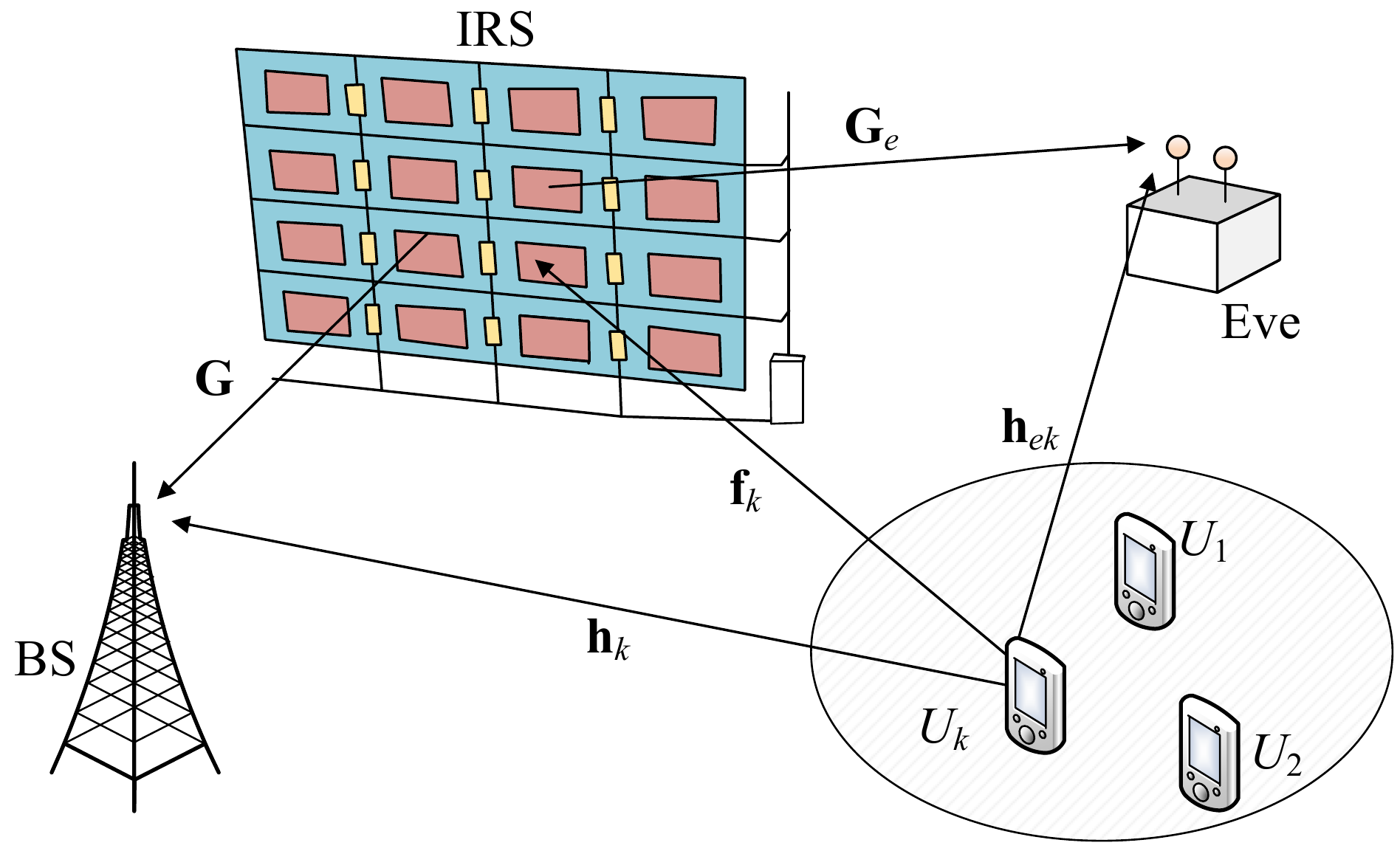}
\caption{The IRS-assisted multiple user uplink communication model with $K$ single-antenna users, an $N_s$-element IRS, an $N_t$-antenna BS, and an $N_e$-antenna eavesdropper.}  \label{fig:system model}
\end{figure}

This paper considers an IRS-assisted multiple-user uplink channel as shown in Fig. \ref{fig:system model}, including a base station (BS) with $N_t$ antennas, an eavesdropper (Eve) with $N_e$ antennas, an IRS equipped with $N_s$ programmable phase shift elements, and $K$ legitimate users with single-antenna defined as $\mathcal{U}=\{U_1,...U_K\}$. Assume that wiretap channels obey spatial-uncorrelated Rayleigh fading, i.e., the channel from $U_k$ to Eve is defined as $\mathbf{h}_{e,k}\sim \mathcal{CN}_{N_e,1}(\mathbf{0},\mathbf{I}_{N_e})$, and the channel from IRS to Eve is defined as $\mathbf{G}_{e}\sim \mathcal{CN}_{N_e,N_s}(\mathbf{0}, \mathbf{I}_{N_e}\otimes \mathbf{I}_{N_s})$. The channel from $U_k$ to BS is defined as $\mathbf{h}_k\in\mathbb{C}^{N_t \times 1}$, the channel from IRS to BS is defined as $\mathbf{G}\in\mathbb{C}^{N_t\times N_s}$, and the channel from $U_k$ to IRS is defined as $\mathbf{f}_{k}\in\mathbb{C}^{N_s\times 1}$. $\mathbf{h}_k$, $\mathbf{G}$, and $\mathbf{f}_{k}$ are assumed to be estimated perfectly. IRS controls programmable phase shift elements via a phase shift matrix, where the phase shift matrix is defined as an $N_s\times N_s$ matrix $\bm{\Phi}$, i.e.,
\begin{flalign}
\bm{\Phi}=\text{diag}[\exp(j\theta_1),...,\exp(j\theta_{n}),..., \exp(j\theta_{N_s})],
\end{flalign}
and $\theta_n\in [0,2\pi)$ is the phase introduced by the $n$-th phase shifter element of IRS.

In the uplink channels, with the phase shift matrix $\bm{\Phi}$, the received signals at BS and Eve can be expressed as
\begin{flalign}\label{Le1}
& \mathbf{y} = (\mathbf{H}+\mathbf{G}\bm{\Phi}\mathbf{F})\mathbf{x}+\mathbf{n} \notag                                                      \\
& = (\mathbf{h}_k+\mathbf{G}\bm{\Phi}\mathbf{f}_k)x_k + \sum_{i=1,i\neq k}^K (\mathbf{h}_i+\mathbf{G}\bm{\Phi}\mathbf{f}_i)x_i+\mathbf{n}, \\
& \mathbf{y}_e=(\mathbf{H}_e+\mathbf{G}_e\bm{\Phi}\mathbf{F})\mathbf{x}+\mathbf{n}_e, \label{Le2}
\end{flalign}
where $x_k$ is the confidential information-bearing signal of $U_k$ with the power constraint $\mathbb{E}(|x_k|^2)=\rho_k$, and $\rho_k$ is the fixed power of $U_k$. $\mathbf{n}_b$ and $\mathbf{n}_e$ is the additive white Gaussian noise (AWGN) obeying $\mathcal{CN}_{N_t,1}(0,\sigma^2_b\mathbf{I}_{N_t})$ and $\mathcal{CN}_{N_e,1}(\mathbf{0},\sigma_e^2\mathbf{I}_{N_e})$, respectively. $\mathbf{H}$ is an extended channel matrix that combines $\mathbf{h}_k, \forall k$ by column, so do $\mathbf{F}$ and $\mathbf{H}_e$. By following the channel model, the secrecy rate of $U_k$ is formulated as
\begin{flalign}
C_{s,k}=(C_{m,k}-C_{w,k})^{+}.
\end{flalign}
$C_{m,k}$ and $C_{w,k}$ are defined as the main channel capacity and wiretap channel capacity,
\begin{flalign}
& C_{m,k}=\log_2(1+\text{SINR}_k), \label{formula:cm}                                                                                                                                                                                            \\
& \text{SINR}_k=\frac{\rho_k| \mathbf{w}_k^{\rm{H}}(\mathbf{h}_k+\mathbf{G}\bm{\Phi}\mathbf{f}_k)|^2}{\mathbf{w}_k^{\rm{H}}(\tilde{\mathbf{K}}_k\mathbf{P}\tilde{\mathbf{K}}_k^{\rm{H}}+\sigma_b^2\mathbf{I}_{N_t})\mathbf{w}_k}, \notag \\
& \mathbf{P}=\text{diag}(\rho_1,...,\rho_{k-1},\rho_{k+1},...,\rho_K), \notag                                                                                                                                                            \\
& C_{w,k}=\log_2\bigg(1+\frac{\rho_k}{\sigma^2_e}|( \mathbf{h}_{e,k}+\mathbf{G}_e\bm{\Phi}\mathbf{f}_k)|^2 \bigg),
\end{flalign}
where $\tilde{\mathbf{K}}_k=[\mathbf{h}_1+\mathbf{G}\bm{\Phi}\mathbf{f}_1,...,\mathbf{h}_{k-1}+\mathbf{G}\bm{\Phi}\mathbf{f}_{k-1},\mathbf{h}_{k+1}+\mathbf{G}\bm{\Phi}\mathbf{f}_{k+1},...,\mathbf{h}_K+\mathbf{G}\bm{\Phi}\mathbf{f}_K]$ is an extended channel matrix that excludes $\mathbf{h}_{k}+\mathbf{G}\bm{\Phi}\mathbf{f}_{k}$. $C_{m,k}$ is achievable at BS when using the receiving vector $\mathbf{w}_{k}^{\rm{H}}$. The pessimistic condition is assumed that Eve has enough computing capacity to eliminate inter-user interference and gather the power gain of desired signals, such as using joint minimum mean square error (MMSE) and successive interference cancellation (SIC) scheme \cite{Tse2005}.
\addtolength{\topmargin}{0.021in}
\subsection{Secrecy Outage Probability}
However, the secrecy rate can not be measured due to the lack of $\mathbf{H}_e$ and $\mathbf{G}_e$. In this case, PLS usually uses secrecy outage probability as the security metric for PLS coding or optimization algorithms. The secrecy outage probability is defined as the probability that the targeted PLS coding rate of $U_k$ encoder, i.e., $R_{k}$, is larger than the secrecy rate $C_{s,k}$. From~\cite[Eq. (38)]{Liu2021}, the secrecy outage probability can be expressed  as
\begin{flalign}\label{uso}
P_{k}(R_k) & =P(C_{s,k} \leq R_k \big| \text{transmission}) \notag \\
& =P(C_{w,k}\geq C_{m,k}- R_k).
\end{flalign}
Obviously, the priority of solving this problem is to find the expression of $P_{k}(R_k)$, which is formulated as follows.
\begin{theorem}[Expression of secrecy outage probability]
The secrecy outage probability of $R_k$ in $U_k$, i.e., $P_{k}(R_k)$, is expressed as
\begin{flalign}\label{esop}
P_{k}(R_k)=\frac{1}{\Gamma(N_e)}\Gamma\bigg(N_e,\frac{\phi_k}{1+|\bm{\Phi}\mathbf{f}_k|^2}\bigg),
\end{flalign}
where $\phi_k=\sigma_e^2(2^{C_{m,k}-R_k}-1)/\rho_k$,  $C_{m,k}$ is defined in Eq. (\ref{formula:cm}), $\Gamma(x)$ is the Gamma function with respect to $x$, and $\Gamma(\epsilon, \eta)$ is the upper incomplete Gamma function defined as follows,
\begin{flalign}\label{ligamma}
\Gamma(\epsilon,\eta)=\int_{\eta}^{\infty}\exp(-z)z^{\epsilon-1}\text{d}z,
\end{flalign}
\end{theorem}

\begin{proof}
See \cite[Corollary 2]{Liu2022}.
\end{proof}

\subsection{Problem Formulation}
To reduce the secrecy outage probability and guarantee user fairness, the optimization problem is to minimize the maximal secrecy outage probability (the min-max fairness) among the $K$ users, where the phase shift matrix and receiving matrix are jointly optimized, which is formulated as follows.
\begin{flalign}
\text{P1: }  \min_{\bm{\Phi},\mathbf{W}} & \max_{k}P_{k}(R_k),\\
 \text{s.t. } & |\exp(j\theta_n)|^2=1, \quad n=1,...,N_s, \label{P1C1}  \\
 &\mathbf{w}_k^{\rm{H}}\mathbf{w}_k = 1, k=1,...,K, \label{P1C2}  
\end{flalign}
where $P_{k}(R_k)$ can be found in Eq. (\ref{esop}), $\bm{\Phi}$ is the phase shift matrix, and $\mathbf{W}\in \mathbb{C}^{K\times N_t}$ is the receiving matrix that combines $\mathbf{w}_k^{\rm{H}},\forall k$ by row. Eq. (\ref{P1C1}) corresponds to the unit-modulus requirements of the reflection elements at the IRS and Eq. (\ref{P1C2}) is the receiving vector constraint.

\section{Alternating Optimization for Min-Max Secrecy Outage Probability}   \label{sec:Opt}
According to Eq. (\ref{esop}), we use an auxiliary value $z_k$ as 
\begin{flalign}
& z_k = \frac{\phi_k}{(1+|\bm{\Phi}\mathbf{f}_k|^2)}. \label{zk}
\end{flalign}
Since $N_e>0$ and $z_k>0$, we have
\begin{flalign}
\frac{\partial \Gamma(N_e,z_k)}{\partial z_k}=-z_k^{N_e-1}\exp(-z_k),
\end{flalign}
and it is concluded that the optimization objective $P_{k}(R_k)$ of P1 decreases with an increasing $z_k$. Hence, we can transform P1 into P2 as follows,
\begin{flalign}
\text{P2: } & \max_{\bm{\Phi},\mathbf{W}}\min_{k}z_k,      \\
& \text{s.t. }  \text{Eqs.} (\ref{P1C1})\text{ and }(\ref{P1C2}).
\end{flalign}
Obviously, it is difficult to optimize $\bm{\Phi}$ and $\mathbf{W}$ as they are coupled, so we transform P2 into two subproblems as P3 and P4, i.e., the optimization problems of receiving matrix and phase shift matrix as follows,
\begin{flalign}
\text{P3: } & \max_{\mathbf{W}}\min_{k}z_k, \notag\\
& \text{  s.t. }  \text{Eq.} (\ref{P1C2}). \\
\text{P4: } & \max_{\bm{\Phi}}\min_{k}z_k, \notag\\
& \text{  s.t. }  \text{Eq.} (\ref{P1C1}).
\end{flalign}
In Section \ref{sbsec:recv_vector_opt}, we solve P3 using the generalized Rayleigh quotient. P4 is solved by the SDR approach and generalized Dinkelbach’s algorithm in Section \ref{sbsec:phi_opt}. Then, we use an AO algorithm to find the global results for $\bm{\Phi}$ and $\mathbf{W}$ in Section \ref{sbsec:AO}.

\subsection{Receiving Vector Optimization}  \label{sbsec:recv_vector_opt}
For any given phase shift matrix $\bm{\Phi}$, the secrecy outage probability of each user is independent of others. Hence, the optimization problem of each user is independent. In this case, P3 can be equivalently transformed into P5 as follows,
\begin{flalign}
\text{P5: } & \min_{k}\max_{\mathbf{W}}z_k,  \label{min_max_zk} \\
& \text{s.t. } \mathbf{w}_k^{\rm{H}}\mathbf{w}_k = 1, k=1,...,K.
\end{flalign}
To solve P5, $z_k$ in Eq. (\ref{min_max_zk}) can be transformed into Eq. (\ref{zk_trans}) as follows,
\begin{flalign}
z_k & = \frac{\phi_k}{(1+|\bm{\Phi}\mathbf{f}_k|^2)} \notag                                                                                                                                                                                                                                             \\
& =c_{1k}\bigg( \frac{| \mathbf{w}_k^{\rm{H}}(\mathbf{h}_k+\mathbf{G}\bm{\Phi}\mathbf{f}_k)|^2}{\mathbf{w}_k^{\rm{H}}(\tilde{\mathbf{K}}_k\mathbf{P}\tilde{\mathbf{K}}_k^{\rm{H}}+\sigma_b^2\mathbf{I}_{N_t})\mathbf{w}_k} \bigg)+c_{2k} \notag                                                               \\
& =c_{1k}\bigg( \frac{\mathbf{w}_k^{\rm{H}}(\mathbf{h}_k+\mathbf{G}\bm{\Phi}\mathbf{f}_k)(\mathbf{h}_k+\mathbf{G}\bm{\Phi}\mathbf{f}_k)^{\rm{H}}\mathbf{w}_k}{\mathbf{w}_k^{\rm{H}}(\tilde{\mathbf{K}}_k\mathbf{P}\tilde{\mathbf{K}}_k^{\rm{H}}+\sigma_b^2\mathbf{I}_{N_t})\mathbf{w}_k} \bigg)+c_{2k} \notag \\
& =c_{1k}\bigg( \frac{\mathbf{w}_k^{\rm{H}}\mathbf{A}_k\mathbf{w}_k}{\mathbf{w}_k^{\rm{H}}\mathbf{B}_k\mathbf{w}_k} \bigg)+c_{2k}.    \label{zk_trans}
\end{flalign}
where
\begin{flalign}
& c_{1k}=\sigma_e^2/[2^{R_k}(1+|\mathbf{f}_k|^2)],   \label{c1_def}                                                                                    \\
& c_{2k}=\sigma_e^2(1-2^{R_k})/[\rho_k2^{R_k}(1+|\mathbf{f}_k|^2)],  \label{c2_def}                                                                     \\
& \mathbf{A}_k = (\mathbf{h}_k+\mathbf{G}\bm{\Phi}\mathbf{f}_k)(\mathbf{h}_k+\mathbf{G}\bm{\Phi}\mathbf{f}_k)^{\rm{H}}, \label{A_def} \\
& \mathbf{B}_k = (\tilde{\mathbf{K}}_k\mathbf{P}\tilde{\mathbf{K}}_k^{\rm{H}}+\sigma_b^2\mathbf{I}_{N_t}). \label{B_def}
\end{flalign}
According to Eq. (\ref{zk_trans}), we can transform P5 into P6 as follows,
\begin{flalign}
\text{P6: } & \min_{k}\max_{\mathbf{W}} c_{1k}\bigg( \frac{\mathbf{w}_k^{\rm{H}}\mathbf{A}_k\mathbf{w}_k}{\mathbf{w}_k^{\rm{H}}\mathbf{B}_k\mathbf{w}_k} \bigg)+c_{2k}. \\
&\text{s.t. } \mathbf{w}_k^{\rm{H}}\mathbf{w}_k = 1, k=1,...,K. 
\end{flalign}
For each user $k$, the optimal receiving vector, i.e., $\mathbf{w}_{k}^{*}$ can be found by solving the inner maximization problem of P6 as follows,
\begin{flalign}
&\mathbf{w}_{k}^{*} = \arg\max_{\mathbf{w}_k} c_{1k}\bigg( \frac{\mathbf{w}_k^{\rm{H}}\mathbf{A}_k\mathbf{w}_k}{\mathbf{w}_k^{\rm{H}}\mathbf{B}_k\mathbf{w}_k} \bigg)+c_{2k}.   \label{wk:argmaxzk}
\end{flalign}
It is obvious that $\mathbf{A}_k$ and $\mathbf{B}_k$ are Hermitian matrices and $\mathbf{B}_k\succ \mathbf{0}$ according to Eqs (\ref{A_def}) and (\ref{B_def}). Thus, with the property of generalized Rayleigh quotient, the optimal $\mathbf{w}_{k}$ can be derived as follows \cite[Pro 2.1.1]{absil2009optimization},
\begin{flalign}
&\mathbf{w}_{k}^{*} = \text{eigvec}_{\lambda_{max}}(\mathbf{B}_k^{-1}\mathbf{A}_k).  \label{optimal_w_k}
\end{flalign}
where $\text{eigvec}_{\lambda_{max}}(\mathbf{X})$ means the corresponding eigenvector of the largest eigenvalue of matrix $\mathbf{X}$, and $\lambda_{max}$ is the largest eigenvalue of $\mathbf{X}$. For each user, we get the optimal receiving vector by Eq. (\ref{optimal_w_k}), then we have the receiving matrix $\mathbf{W}$ for fixed phase shift matrix $\bm{\Phi}$. It can be said that the subproblem P3 is solved optimally.

\subsection{Phase Shift Matrix Optimization} \label{sbsec:phi_opt}
In the IRS-assisted multiple-user model, the change of phase shift matrix $\bm{\Phi}$ of IRS affects all users. Thus, we cannot consider this problem to be independent for each user as we did with P3. According to the Eq. (\ref{zk_trans}), the P4 can be transformed into P7 as 
\begin{flalign}
\text{P7: } & \max_{\bm{\Phi}}\min_{k}  c_{1k}\bigg(   \frac{| \mathbf{w}_k^{\rm{H}}(\mathbf{h}_k+\mathbf{G}\bm{\Phi}\mathbf{f}_k)|^2}{\mathbf{w}_k^{\rm{H}}(\tilde{\mathbf{K}}_k\mathbf{P}\tilde{\mathbf{K}}_k^{\rm{H}}+\sigma_b^2\mathbf{I}_{N_t})\mathbf{w}_k}  \bigg)+c_{2k}, \\
& \text{  s.t. }  \text{Eq.} (\ref{P1C1}).   \label{P7C1}
\end{flalign}
In order to solve P7 conveniently, we make the following mathematical transformations. The phase shift matrix $\bm{\Phi}$ is a diagonal matrix and $\mathbf{f}_{k}$ is a vector, so $\mathbf{G}\bm{\Phi}\mathbf{f}_k$ can be transformed as follows,
\begin{flalign}
\mathbf{G}\bm{\Phi}\mathbf{f}_k & = \mathbf{G}\text{diag}(\mathbf{f}_k)\text{vec}(\mathbf{\Phi}) =\mathbf{E}_k\mathbf{q},   \label{gphif_trans_eq}
\end{flalign}
where
\begin{flalign}
&\mathbf{q} = \text{vec}(\mathbf{\Phi}),     \quad \mathbf{E}_k = \mathbf{G}\text{diag}(\mathbf{f}_k).
\end{flalign}
The denominator of the optimization objective of P7 can be transformed as follows,
\begin{flalign}  \label{denominator_trans}  
&\mathbf{w}_k^{\rm{H}}(\tilde{\mathbf{K}}_k\mathbf{P}\tilde{\mathbf{K}}_k^{\rm{H}}+\sigma_b^2\mathbf{I}_{N_t})\mathbf{w}_k \notag \\
&=  \mathbf{w}_k^{\rm{H}}\tilde{\mathbf{K}}_k\mathbf{P}\tilde{\mathbf{K}}_k^{\rm{H}}\mathbf{w}_k + t \notag \\
&=  \sum_{i=1,i\neq k}^{K} \rho_i|(\mathbf{h}_i+\mathbf{E}_i\mathbf{q})^{\rm{H}}\mathbf{w}_k|^2 + t \notag \\
& =\sum_{i=1,i\neq k}^{K}  \rho_i(\mathbf{h}_i^{\rm{H}}\mathbf{w}_k\mathbf{w}_k^{\rm{H}}\mathbf{h}_i+\mathbf{h}_i^{\rm{H}}\mathbf{w}_k\mathbf{w}_k^{\rm{H}}\mathbf{E}_i\mathbf{q}\notag  \\
& + \mathbf{q}^{\rm{H}}\mathbf{E}_i^{\rm{H}}\mathbf{w}_k\mathbf{w}_k^{\rm{H}}\mathbf{h}_i+\mathbf{q}^{\rm{H}}\mathbf{E}_i^{\rm{H}}\mathbf{w}_k\mathbf{w}_k^{\rm{H}}\mathbf{E}_i\mathbf{q})+t \notag   \\
& =\sum_{i=1,i\neq k}^{K}   \rho_i(\hat{\mathbf{q}}^{\rm{H}}\mathbf{M}_{i,k}\hat{\mathbf{q}}  + \mathbf{v}_{i,k}) + t,
\end{flalign}
where
\begin{flalign}
&t = \mathbf{w}_k^{\rm{H}}\sigma_b^2\mathbf{I}_{N_t}\mathbf{w}_k,\notag \\
&\mathbf{M}_{i,k}=\left[
\begin{matrix}
\mathbf{E}_i^{\rm{H}}\mathbf{w}_k\mathbf{w}_k^{\rm{H}}\mathbf{E}_i & \mathbf{E}_i^{\rm{H}}\mathbf{w}_k\mathbf{w}_k^{\rm{H}}\mathbf{h}_i \\
\mathbf{h}_i^{\rm{H}}\mathbf{w}_k\mathbf{w}_k^{\rm{H}}\mathbf{E}_i & 0
\end{matrix}
\right],                                                                                          \label{mik}\\
&\mathbf{v}_{i,k}=\mathbf{h}_i^{\rm{H}}\mathbf{w}_k\mathbf{w}_k^{\rm{H}}\mathbf{h}_i, \label{vik} \\
&\hat{\mathbf{q}}^{\rm{H}}  =\left[
\begin{matrix}
\mathbf{q}^{\rm{H}}, & l
\end{matrix}
\right],  \quad \mathbf{E}_i = \mathbf{G}\text{diag}(\mathbf{f}_i).\label{q_hat_def}
\end{flalign}
and $l$ is an auxiliary variable. Similarly, the numerator of the optimization objective of P7 can be transformed as
\begin{flalign}\label{q1}
| \mathbf{w}_k^{\rm{H}}(\mathbf{h}_k+\mathbf{G}\bm{\Phi}\mathbf{f}_k)|^2=\hat{\mathbf{q}}^{\rm{H}}\mathbf{M}_{k,k}\hat{\mathbf{q}} + \mathbf{v}_{k,k},
\end{flalign}
%where
%\begin{flalign}
%\mathbf{M}_{k,k}&=\left[
%\begin{matrix}
%\mathbf{E}_k^{\rm{H}}\mathbf{w}_k\mathbf{w}_k^{\rm{H}}\mathbf{E}_k & \mathbf{E}_k^{\rm{H}}\mathbf{w}_k\mathbf{w}_k^{\rm{H}}\mathbf{h}_k \\
%\mathbf{h}_k^{\rm{H}}\mathbf{w}_k\mathbf{w}_k^{\rm{H}}\mathbf{E}_k & 0
%\end{matrix}
%\right],                                                                                          \\
%\mathbf{v}_{k,k}&=\mathbf{h}_k^{\rm{H}}\mathbf{w}_k\mathbf{w}_k^{\rm{H}}\mathbf{h}_k.
%\end{flalign}
where $\mathbf{M}_{k,k}$ and $\mathbf{v}_{k,k}$ have the same structure with Eqs. (\ref{mik}) and (\ref{vik}) by replacing subscript $i$ with $k$. According to the Eqs. (\ref{denominator_trans}) and (\ref{q1}), P7 can be transformed into P8 as follows.
\begin{flalign}
\text{P8: }  \max_{\hat{\mathbf{q}}}\min_{k} & c_{1k}\bigg(  \frac{\hat{\mathbf{q}}^{\rm{H}}\mathbf{M}_{k,k}\hat{\mathbf{q}} + \mathbf{v}_{k,k}}{ \sum_{i=1,i\neq k}^{K}   \rho_i(\hat{\mathbf{q}}^{\rm{H}}\mathbf{M}_{i,k}\hat{\mathbf{q}}  + \mathbf{v}_{i,k}) + t }  \bigg)+c_{2k}        \label{P8P}\\
&\text{s.t. } |\hat{q_i}|^2=1, \forall i = 1,...,N_s+1, \label{P8C1}
\end{flalign}
where $\hat{q_i}$ is the $i$-th element in $\hat{\mathbf{q}}$. Notice that $\hat{\mathbf{q}}^{\rm{H}}\mathbf{M}_{i,k}\hat{\mathbf{q}}$ is a quadric form. By exploiting the fact that $\text{tr}(\mathbf{AB}) = \text{tr}(\mathbf{BA})$,
$\hat{\mathbf{q}}^{\rm{H}}\mathbf{M}_{i,k}\hat{\mathbf{q}}$ can be transformed as follows,
\begin{flalign}
\hat{\mathbf{q}}^{\rm{H}}\mathbf{M}_{i,k}\hat{\mathbf{q}}
& = \text{tr}(\hat{\mathbf{q}}^{\rm{H}}\mathbf{M}_{i,k}\hat{\mathbf{q}})  \notag \\
& = \text{tr}(\mathbf{M}_{i,k}\hat{\mathbf{q}}\hat{\mathbf{q}}^{\rm{H}})  \notag \\
& = \text{tr}(\mathbf{M}_{i,k}\mathbf{Q}),
\end{flalign}
where $\mathbf{Q}= \hat{\mathbf{q}}\hat{\mathbf{q}}^{\rm{H}}$. Also, we have $\hat{\mathbf{q}}^{\rm{H}}\mathbf{M}_{k,k}\hat{\mathbf{q}}=\text{tr}(\mathbf{M}_{k,k}\mathbf{Q})$. To address constraint (\ref{P8C1}), we introduce $N_s+1$ auxiliary matrices $\mathbf{E}_n,n=1,...,N_s+1$ as follows,
\begin{equation}
[\mathbf{E}_n]_{i,j}=
\begin{cases}
1, & \mbox{$i=j=n$},   \\
0, & \mbox{otherwise},
\end{cases} \label{en_def}
\end{equation}
where $[\mathbf{E}_n]_{i,j}$ is the $(i,j)-$th element of $\mathbf{E}_n$. With $\mathbf{E}_n,\forall n$, the constraint (\ref{P8C1}) can be transformed as $\hat{\mathbf{q}}^{\rm{H}}\mathbf{E}_n\hat{\mathbf{q}}=\text{trace}(\mathbf{E}_n\mathbf{Q})=1, \forall n = 1,...,N_s+1$. Then, P8 can be transformed equivalently as
\begin{flalign}
\text{P9: }   \max_{\mathbf{Q}}\min_{k} &  \frac{ c_{1k}(\text{tr}(\mathbf{M}_{k,k}\mathbf{Q}) + \mathbf{v}_{k,k}) } { \sum_{i!=k}^{K}  \rho_i(\text{tr}(\mathbf{M}_{i,k}\mathbf{Q}) + \mathbf{v}_{i,k})  + t  } + c_{2k},  \\
\text{s.t. } & \text{rank}(\mathbf{Q}) = 1,         \label{P9C1:rankone}  \\
& \text{tr}(\mathbf{E}_n\mathbf{Q})=1, \forall n = 1,...,N_s+1,\label{P9C2:enxq_eq_1} \\
& \mathbf{Q}\succeq0,  \label{P9C3:sem_Q}
\end{flalign}
where the constraints in P9 are convex except the $\text{rank}(\mathbf{Q}) = 1$ is a non-convex constraint. Thus, we apply semi-definite relaxation to relax this constraint to get a convex problem. P9 without the rank-one constraint can be regarded as an optimization problem that maximizes the ratios in which both the numerator and the denominator are affine, and generalized Dinkelbach's algorithm is suitable to solve it. Generalized Dinkelbach's algorithm belongs to the class of parametric algorithms, which can decompose P9 into a sequence of convex problems and obtain the global solution at a linear convergence rate\cite{FracProgTheory_Zappone_2015}.

In order to execute the generalized Dinkelbach's algorithm conveniently,  we denote the numerator and denominator of the optimization objective in P9 as $N_k(\mathbf{Q})$ and $D_k(\mathbf{Q})$,
\begin{flalign}
&N_k(\mathbf{Q}) = c_{1k}(\text{tr}(\mathbf{M}_{k,k}\mathbf{Q}) + \mathbf{v}_{k,k}) + c_{2k}D_k(\mathbf{Q}),  \label{eq:nkq}\\
&D_k(\mathbf{Q})= \sum_{i!=k}^{K}  \rho_i(\text{tr}(\mathbf{M}_{i,k}\mathbf{Q}) + \mathbf{v}_{i,k})  + t.
\end{flalign}
The constraints of P9 are denoted by a set $\mathcal{S}$ except for the rank-one constraint in Eq. (\ref{P9C1:rankone}) as follows,
\begin{flalign}
\mathcal{S} = \{ \mathbf{Q} \text{ } | \text{ Eqs. (\ref{P9C2:enxq_eq_1}), (\ref{P9C3:sem_Q})\}}.
\end{flalign}

\begin{algorithm}[t]  
	\DontPrintSemicolon
	\KwInput{$\mathbf{G, H, F, P}, \sigma_b, N_t, N_s, K,R_k$}
	\KwOutput{$\mathbf{Q^{*}}$}
	Initialize $\tau > 0, \lambda = 0$; \;
	\Do{$F > \tau $ }
	{$ \mathbf{Q}^{*} = \arg\max\limits_{\mathbf{Q} \in \mathcal{S}} \left\{ \min\limits_{1 \leq k \leq K} \{N_k(\mathbf{Q}) - \lambda D_k(\mathbf{Q})\}  \right\}$;  \;
		$F = \min\limits_{1 \leq k \leq K} \{N_k(\mathbf{Q}^{*}) - \lambda D_k(\mathbf{Q}^{*})\} $; \;
		$\lambda = \min\limits_{1 \leq k \leq K} \frac{N_k(\mathbf{Q}^{*})}{D_k(\mathbf{Q}^{*})}$; \;
		}
	Get result $\mathbf{Q^{*}} $; \;
	\textbf{Procedure End}
	\caption{Generalized Dinkelbach's Method for P9 Except Rank-One Constraint.}
	\label{alg:dink}
\end{algorithm}
The procedure of solving P9 by generalized Dinkelbach's algorithm is described in Algorithm \ref{alg:dink}, where $\tau$ is an arbitrarily small value. In detail, the step 3 in Algorithm \ref{alg:dink} can be transformed by introducing an auxiliary value $u$,
\begin{flalign}
\text{P10: }  & \max\limits_{\mathbf{Q}}  u,  \\
\text{s.t. }  & \mathbf{Q} \in \mathcal{S}, \\
& u \leq N_k(\mathbf{Q}) - \lambda_nD_k(\mathbf{Q}), \forall k = 1,...,K.
\end{flalign}
P10 is a convex problem, which can be solved by interior-point methods or MATLAB CVX tool.

By generalized Dinkelbach's algorithm, we can obtain $\mathbf{Q^{*}}$. If $\text{rank}(\mathbf{Q^{*}}) = 1$ , the optimal $\hat{\mathbf{q}}$, i.e., $\hat{\mathbf{q}}^{*}$ can be obtained as follows,
\begin{flalign}
\hat{\mathbf{q}}^{*} = \text{eigvec}_{\lambda_{max}}(\mathbf{Q^{*}}).
\end{flalign}
However, when $\text{rank}(\mathbf{Q^{*}})$ is not equal to one, the near-optimal $\hat{\mathbf{q}}$ can be obtained by maximum eigenvalue methods or the Gaussian randomization method\cite{WuIRS2019}. Here, we choose the Gaussian randomization method to obtain near-optimal $\hat{\mathbf{q}}$.

\subsection{Alternating Optimization}    \label{sbsec:AO}
% In order to address P2, we transform P2 into two subproblems P3 and P4, representing the optimization of the receiving matrix and the phase shift matrix, respectively. 
In previous sections, we use the generalized Rayleigh quotient and generalized Dinkelbach's algorithm to solve P3 and P4, respectively. Here, we design an AO algorithm for the global results of the receiving matrix and phase shift matrix, as shown in Algorithm \ref{alg:ao}. 

\begin{algorithm}[h]
\DontPrintSemicolon
\KwInput{$\mathbf{G, H, F, P, R}, \sigma_b, \sigma_e, N_t, N_s, N_e,R_k$}
\KwOutput{$\mathbf{\Phi}_o^*,\mathbf{W}_o^*$}
Initialize iter = 1 and the iterating limit is $\text{iter}_\text{max}$; \;
Initialize receiving matrix $\mathbf{Q}^*_0$ randomly; \;
\Do{$\text{iter} < \text{iter}_\text{max} \text{\&\&}  \{\text{P}_{\text{out}}(\text{iter}) -\text{P}_{\text{out}}(\text{iter}-1) \leq \xi \} $}
{Solve P3 via generalized Rayleigh quotient to obtain $\mathbf{W}^*_{\text{iter}}$; \;
Obtain $\mathbf{Q}_\text{iter}^*$ with fixed $\mathbf{W}^*_{\text{iter}}$ by generalized Dinkelbach's algorithm; \;
Calculate $P_{k}(R_k),\forall k$ via Eq. (\ref{esop}) and takes the maximum value for each user, i.e., $P_{\text{out}}(\text{iter})$; \;
Update iter = iter + 1;}
$\mathbf{Q}^*_\text{o} = \mathbf{Q}^*_{\text{iter}}, \mathbf{W}^*_\text{o} = \mathbf{W}^*_{\text{iter}}$;\;
Get $\hat{\mathbf{q}}^*_\text{o}$ via Gaussian randomization method of $\mathbf{Q}^*_\text{o}$;\\
\textbf{Procedure End}
\caption{AO Algorithm for Receiving Matrix and Phase Shift Matrix.}
\label{alg:ao}
\end{algorithm} 

In Algorithm \ref{alg:ao}, $\mathbf{Q}_\text{iter}^*$ and $\mathbf{W}^*_{\text{iter}}$ are local results with in the $\text{iter}$-th iteration, and $P_{\text{out}}(\text{iter})$ is the min-max secrecy outage probability in the $\text{iter}$-th iteration. We set an arbitrarily small value $\xi$, where the iteration process continues until $P_{\text{out}}(\text{iter}) -P_{\text{out}}(\text{iter}-1) \leq \xi  $. With the iteration, the receiving matrix and phase shift matrix are alternately optimized, and $P_{\text{out}}(\text{iter})$ gradually converges to constants, as do $\mathbf{Q}_\text{iter}^*$ and $\mathbf{W}^*_{\text{iter}}$. Then, the iteration terminated. To avoid the endless loop, we set $\text{iter}_\text{max}$ as the maximum number of allowed loops.

\subsection{Computational Complexity Analysis}
The computational complexity of Algorithm \ref{alg:ao} is related to the number of iterations and the complexity of solving two subproblems in steps 4, 5, 6, 10, respectively. According to Eq. (\ref{optimal_w_k}), we know that step 4 requires $K(3N_t^3 + KN_t^2)$ iterations. Step 5 exploits Algorithm \ref{alg:dink} to obtain the phase shift matrix with a linear convergence rate. The complexity of each iteration in Algorithm \ref{alg:dink} is on the order of $O[K\ln(1/\epsilon)(KN_s)^{4.5}]$, where $\epsilon$ is the accuracy requirement of the interior-point method. According to Eq. (\ref{esop}), the computational complexity of step 6  is on the order of $O(K^2N_s^3)$. The computational complexity of step 10 is on the order of $O[(N_s)^3 + I_GKN_s^3]$, where $I_G$ denotes the number of iterations of the Gaussian randomization method. In total, the computational complexity of Algorithm \ref{alg:ao} is on the order of $O(I_A(KN_t^3 + K^2N_t^2 + I_DK\ln(1/\epsilon)(KN_s)^{4.5}) + I_GKN_s^3)$, where $I_A \text{ and } I_D$ denote the numbers of iterations of AO algorithm and generalized Dinkelbach's algorithm, respectively.

%%%%%%%仿真%%%%%%%%
\section{SIMULATIONS}   \label{sec:simu}

In this section, simulation results are provided to test the proposed optimization scheme. The performance of the proposed scheme is tested in different $N_t, N_s, \text{and } N_e$, and the max-min SINR scheme \cite{Wang2021,Gong2021MMSINR,Nadeem2020MMSINR} as a comparison for our proposed scheme. The number of users $K$ is set to 4 to simulate a multi-user scenario. The power $\rho_k$ is set to be the same for all users, so is the PLS coding rate $R_k$.

\begin{figure}[htbp]
	\centering
	\includegraphics[width=1\linewidth]{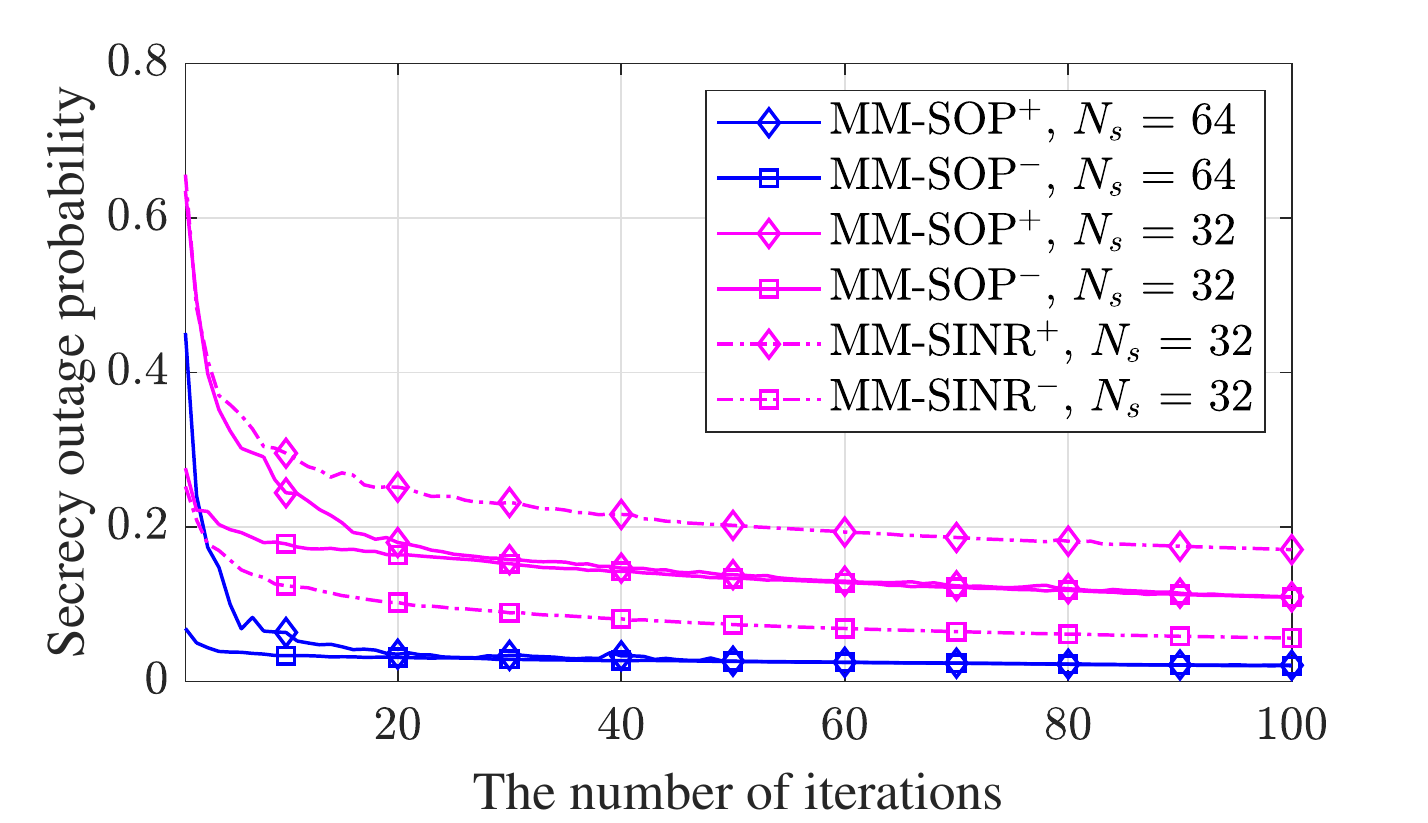}
	\caption{Convergence test of AO algorithm, where $N_t = 10, N_e = 2,$ SNR = 1 dB, and PLS coding rate is 2 bit/s/Hz.}\label{fig:test_iter_conv}
\end{figure}

We first perform a convergence test on the proposed AO algorithm. As shown in Fig. \ref{fig:test_iter_conv}, we conducted convergence tests with different $N_s$. Where $\text{MM-SOP}^+$ and $\text{MM-SOP}^-$ mean the maximum and minimum value in multiple users respectively, with the proposed scheme that minimizes the maximum secrecy outage probability (MM-SOP). MM-SINR means the scheme that maximizes the minimum SINR \cite{Wang2021,Gong2021MMSINR,Nadeem2020MMSINR}. In several independent simulations, our optimization objective, the maximum secrecy outage probability in multi-users, decreases rapidly in the first few iterations and converges to a constant in later iterations. The curve shows a good performance of the proposed AO algorithm. Compared with the MM-SOP scheme, the secrecy outage probability gap between multi-users using the MM-SINR scheme is larger. In addition, the convergence rate is slower in $N_s$ = 32 compared to that of $N_s$ = 64. A larger number of IRS elements is more conducive to the system with the proposed scheme.

\begin{figure*}[htbp]
	\begin{subfigure}[t]{.32\textwidth}
		\centering
		\includegraphics[width=1\textwidth]{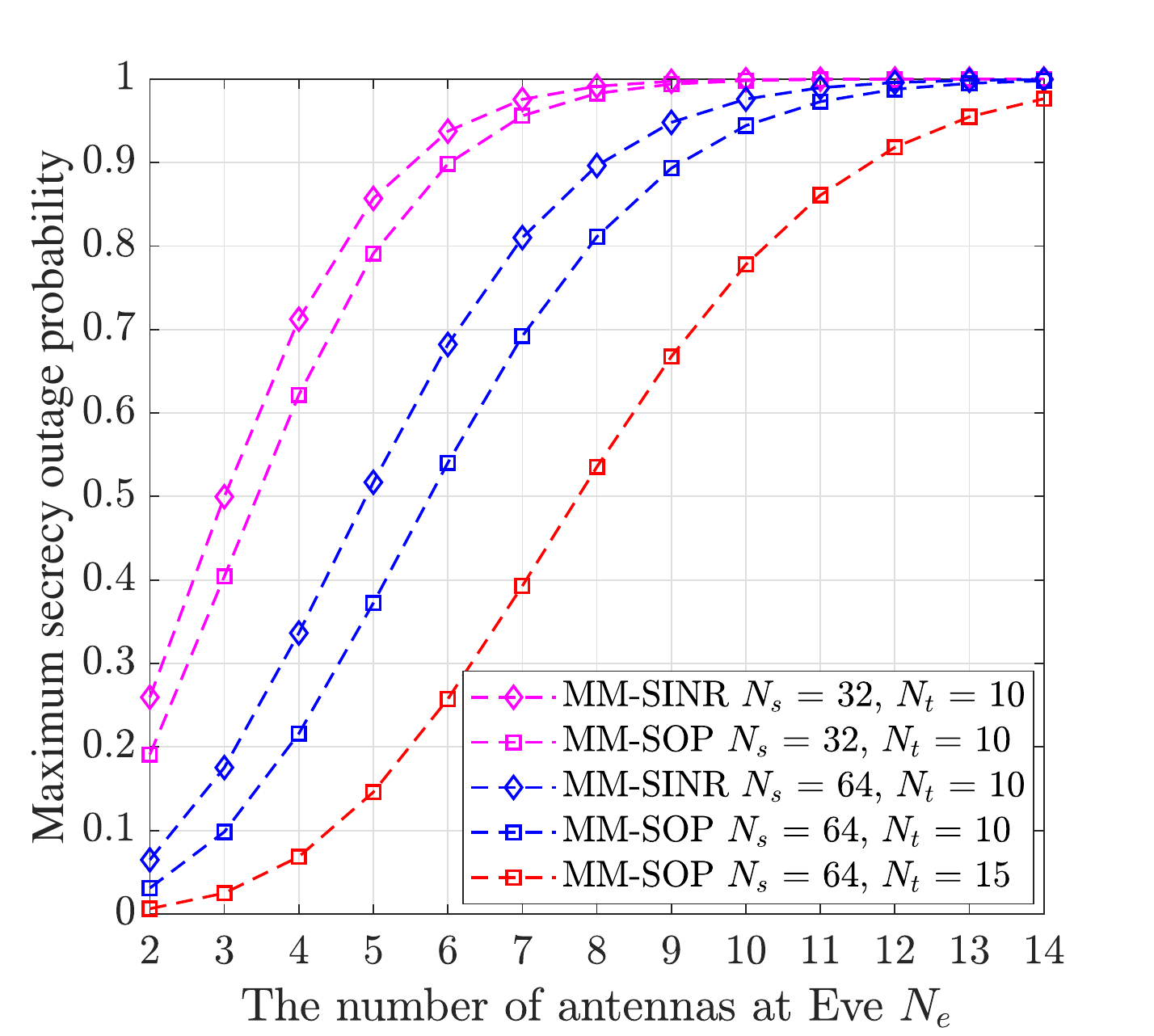}
		\caption{Effect of $N_e$, where SNR = 1 dB and PLS coding rate is 2 bit/s/Hz.}\label{fig:testne}
	\end{subfigure}\hfill
	\begin{subfigure}[t]{.32\textwidth}
		\centering
		\includegraphics[width=1\textwidth]{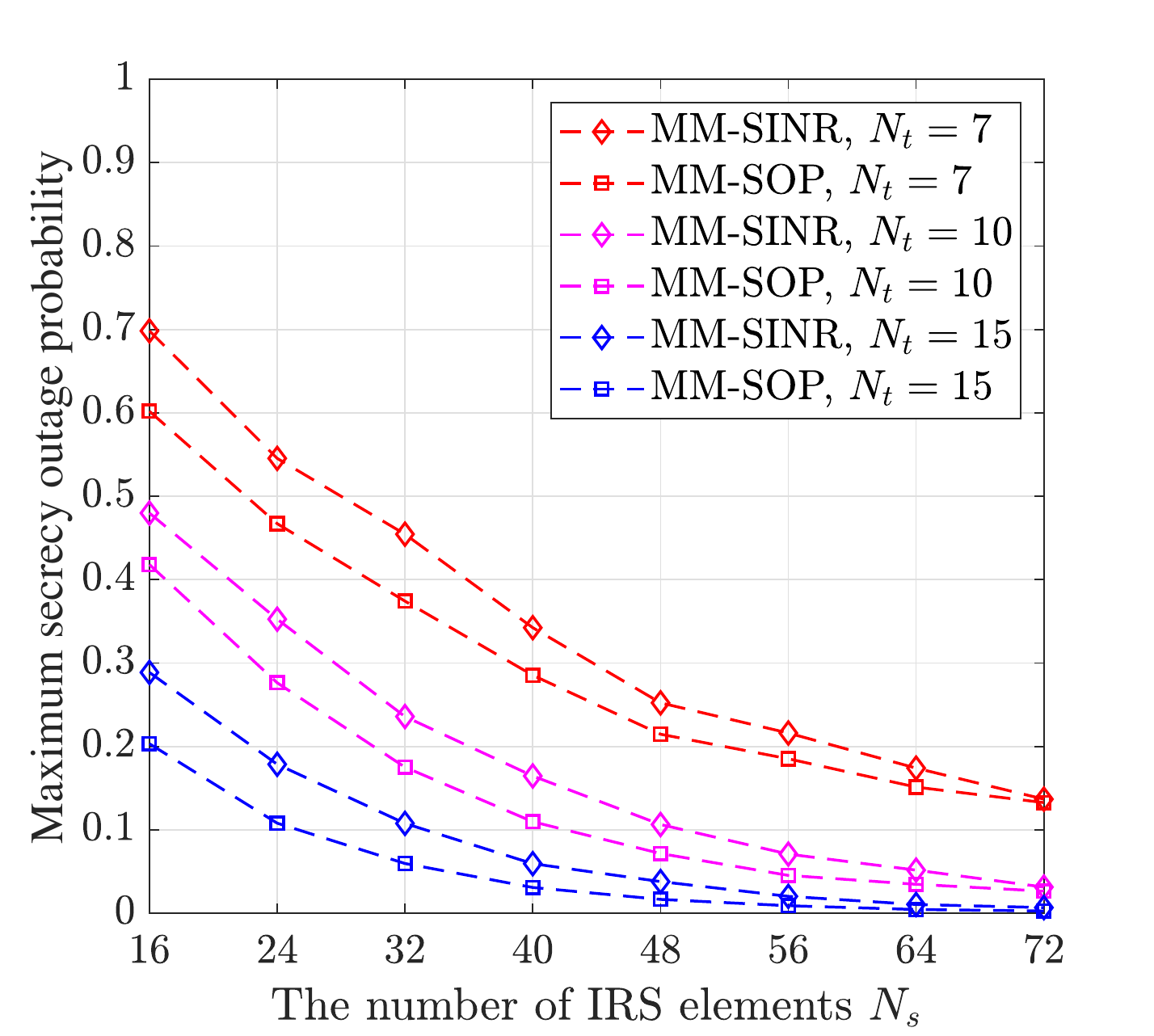}
		\caption{Effect of $N_s$, where $N_e = 2$, SNR = 1 dB, and PLS coding rate is 2 bit/s/Hz.}\label{fig:test_ns}
	\end{subfigure}
	\begin{subfigure}[t]{.32\textwidth}
		\centering
		\includegraphics[width=1\textwidth]{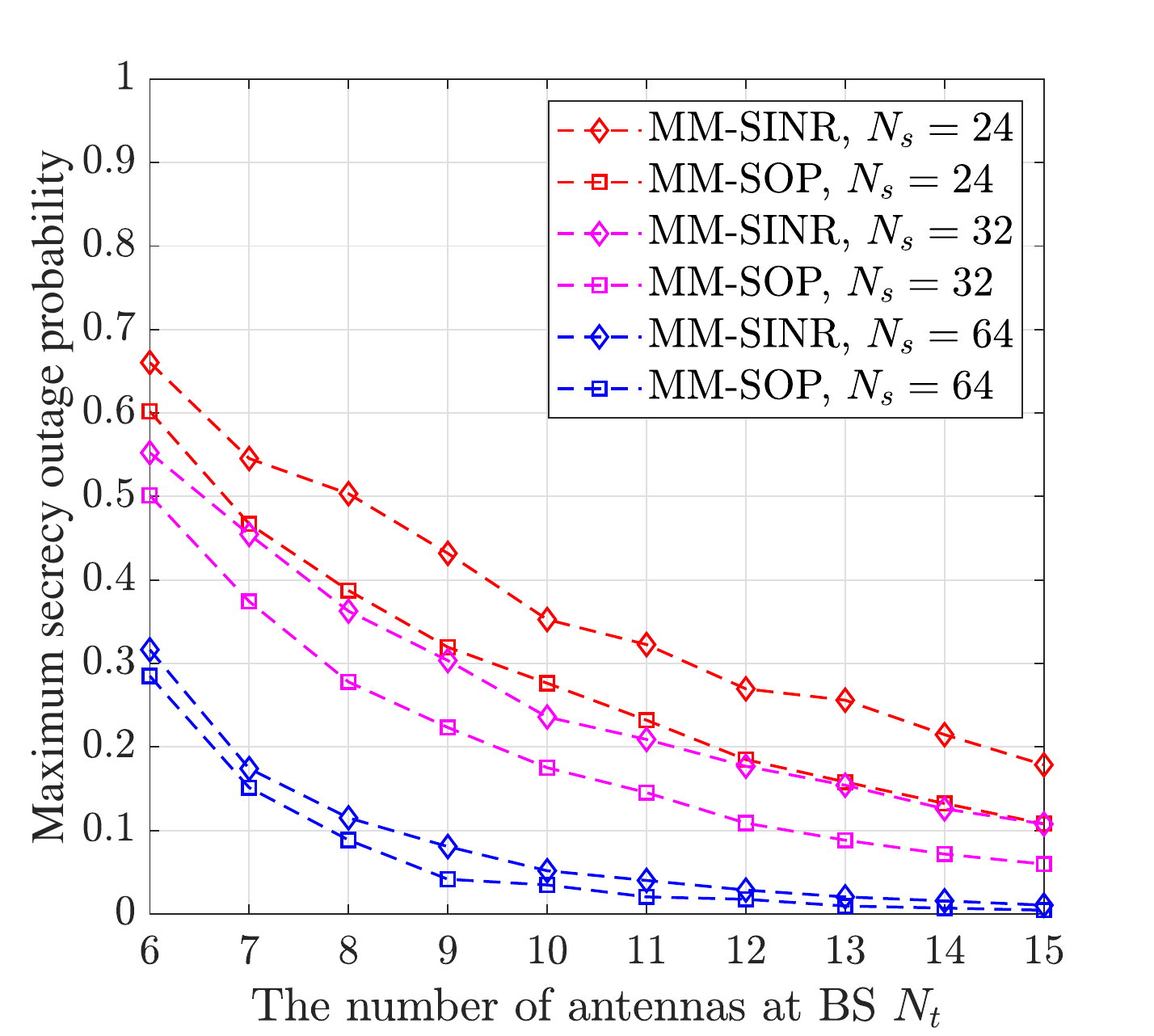}
		\caption{Effect of  $N_t$, where $N_e = 2$, SNR = 1 dB, and PLS coding rate is 2 bit/s/Hz.}\label{fig:test_nt}
	\end{subfigure}
	\caption{Maximum secrecy outage probability among users.}
\end{figure*}

Fig. \ref{fig:testne} shows the impact of $N_e$ on secrecy outage probability. We can find that the secrecy outage probability increases with an increasing $N_e$. In addition, Fig. \ref{fig:testne} also shows that more antennas of BS or more IRS reflecting elements can decrease the maximum secrecy outage probability of users, both for MM-SOP and MM-SINR. Fig. \ref{fig:test_ns} shows in more detail the impact of $N_s$ on the maximum secrecy outage probability. With the increase of $N_s$, the maximum secrecy outage probability gradually decreases. Fortunately, the reflecting elements of IRS are low-cost and can be equipped in large quantities. The impact of $N_t$ is shown in Fig. \ref{fig:test_nt}. It can be observed that more BS antennas can reduce the secrecy outage probability and the system is more sensitive to $N_t$ than to $N_s$. Fig. \ref{fig:testne}, \ref{fig:test_ns}, and \ref{fig:test_nt} also show that the performance of the MM-SOP scheme is superior than that of the MM-SINR scheme.

\section{CONCLUSIONS}   \label{sec:conclusions}
This paper focuses on the PLS without the instantaneous CSI of eavesdroppers in IRS-assisted multiple-user uplink channels. We have proposed an AO algorithm to minimize the maximum outage probability with fairness among multiple users. In detail, receiving vector of each user and phase shift matrix are optimized alternately by generalized Rayleigh quotient and generalized Dinkelbach's algorithm, respectively. Simulation results have shown the excellent convergence performance of the AO algorithm. Moreover, the proposed scheme significantly reduces the maximum secrecy outage probability compared to the max-min SINR method.

\bibliographystyle{IEEEtran}
\bibliography{references}
\end{document}